\documentclass[11pt, reqno]{elsarticle}%
\usepackage{lipsum}
\makeatletter
\def\ps@pprintTitle{%
 \let\@oddhead\@empty
 \let\@evenhead\@empty
 \def\@oddfoot{}%
 \let\@evenfoot\@oddfoot}
\makeatother

\usepackage{amsmath}
\usepackage{bbm}
\usepackage[dvipsnames]{xcolor}

\usepackage{array}
\usepackage{amssymb}
\usepackage{amsthm}
\theoremstyle{plain}

\usepackage{dsfont} 

\usepackage[margin=1in]{geometry}
\usepackage[bookmarks=false]{hyperref}
\usepackage{epstopdf}

\def\+{{\oplus}}

\newcounter{comments}
\newtheorem*{theorem*}{Theorem}
\newtheorem{thm}{Theorem}[section]
\newtheorem{theorem}[thm]{Theorem}
\newtheorem{lemma}[thm]{Lemma}
\newtheorem{cor}[thm]{Corollary}

\newtheorem{defn}[thm]{Definition}
\newtheorem{example}[thm]{Example}
\newtheorem{remark}[thm]{Remark}

\numberwithin{equation}{section}

\begin{document}

\title{On the number of non-degenerate canalizing Boolean functions}

\author[ia]{C. Kadelka\corref{cor1}\fnref{fn1}}
\ead{ckadelka@iastate.edu}

\cortext[cor1]{Corresponding author}
\fntext[fn1]{Supported by NSF Grants DMS-2424632 and DMS-2451973.}

\address[ia]{Department of Mathematics,
      Iowa State University, Ames, IA 50011, USA}


\begin{abstract}
Canalization is a key organizing principle in complex systems, particularly in gene regulatory networks. It describes how certain input variables exert dominant control over a function’s output, thereby imposing hierarchical structure and conferring robustness to perturbations. Degeneracy, in contrast, captures redundancy among input variables and reflects the complete dominance of some variables by others. Both properties influence the stability and dynamics of discrete dynamical systems, yet their combinatorial underpinnings remain incompletely understood. Here, we derive recursive formulas for counting Boolean functions with prescribed numbers of essential variables and given canalizing properties. In particular, we determine the number of non-degenerate canalizing Boolean functions -- that is, functions for which all variables are essential and at least one variable is canalizing. Our approach extends earlier enumeration results on canalizing and nested canalizing functions. It provides a rigorous foundation for quantifying how frequently canalization occurs among random Boolean functions and for assessing its pronounced over-representation in biological network models, where it contributes to both robustness and to the emergence of distinct regulatory roles.
\end{abstract}

\maketitle

\section{Introduction}
Boolean networks provide a fundamental mathematical abstraction for studying discrete dynamical systems across disciplines, from digital circuits and fault-tolerant logic design to gene regulatory and signaling networks~\cite{kauffman1969metabolic}. Each node in such a network is governed by a Boolean function $f: \{0,1\}^n \rightarrow \{0,1\}$ that maps the states of its regulators to its own state, thereby encoding the local regulatory logic and collectively determining the system’s dynamics. 
Despite their simplicity, Boolean networks can exhibit a vast diversity of dynamical behaviors, including multistability, oscillations, and critical transitions~\cite{balleza2008critical,torres2012criticality,schwab2020concepts}.
Understanding how the structural and combinatorial properties of the underlying Boolean functions shape such emergent dynamics remains a central question in the study of complex systems and nonlinear dynamics.

Two important notions in this context are \emph{canalization} and \emph{degeneracy}. 
Originally introduced by Kauffman in his pioneering work on genetic regulatory circuits~\cite{kauffman1974large}, canalization describes the dominance of certain input variables in determining the output of a Boolean function. 
Formally, a function is canalizing if there exists a variable $x_i$ and an input value $a_i \in \{0,1\}$ such that fixing $x_i = a_i$ determines the function's output regardless of the remaining inputs. 
Canalization thus captures hierarchical control among variables, a property that has been linked to robustness in biological and engineered systems~\cite{kauffman2004genetic,kadelka2017influence,kadelka2024canalization}. 
A function is \emph{degenerate} if its output does not depend on all of its input variables -- equivalently, some variables are \emph{non-essential}.
Degeneracy therefore reflects the complete dominance of particular variables over the rest, whose values become irrelevant to the function’s output.

Recent work has established a rigorous mathematical framework for classifying and enumerating canalizing Boolean functions according to their \emph{canalizing depth} and \emph{layer structure}~\cite{kadelka2017influence,layne2012nested,he2016stratification,dimitrova2022revealing}. 
In this framework, a function is $k$-canalizing if there exists an ordered sequence of $k$ variables such that fixing them successively in their canalizing inputs completely determines the function. 
The class of \emph{nested canalizing functions (NCFs)} corresponds to the maximal case $k=n$~\cite{kauffman2003random}. 
Such functions are known to exhibit strong dynamical stability when used as update rules in Boolean networks~\cite{kauffman2004genetic,kadelka2017influence,Shmul04}. 
However, a comprehensive enumeration of canalizing functions that are also \emph{non-degenerate} -- that is, functions that depend on all their inputs -- has not yet been established. 
These functions are of particular theoretical interest because they represent logical rules that are simultaneously fully interactive and hierarchically constrained.

In this work, we derive recursive relations for counting Boolean functions with a specific number of essential variables and a specific canalizing depth. 
Our formulation provides useful recursive expressions for the number of non-degenerate canalizing and nested canalizing functions in $n$ inputs and generalizes prior enumeration results on canalizing and nested canalizing functions~\cite{he2016stratification,dimitrova2022revealing}. 
Beyond its intrinsic combinatorial interest, this analysis provides the quantitative foundation needed to assess how strongly canalization is over-represented in biological regulatory logic. Accurate prevalence estimates are essential for distinguishing genuine design principles from statistical artifacts, enabling a rigorous evaluation of the extent to which canalization and nested canalization are selectively favored in biological network architectures.

\section{Background}\label{sec:background} 
In this section we review the concepts of \emph{degeneracy} and \emph{canalization}. Without loss of generality, we consider Boolean functions defined over the field $\{0,1\}$. 

\begin{defn} \label{def:essential} 
A Boolean function $f(x_{1},\ldots, x_{n})$ is \emph{essential} in the variable $x_{i}$ if 
there exists some $\mathbf x \in \{0,1\}^n$ such that
$$f(\mathbf x) \neq f(\mathbf x \oplus e_i),$$
where $\oplus$ denotes addition modulo $2$ and $e_i$ is the $i$th unit vector. In this case, we call $x_i$ an \emph{essential} or \emph{non-degenerate} variable. Otherwise, $x_i$ is a \emph{non-essential} or \emph{degenerate} variable.
\end{defn}

\begin{defn} \label{def:degenerate} 
A Boolean function $f(x_{1},\ldots, x_{n})$ is \emph{non-degenerate} if it is essential in all its variables. Otherwise, it is \emph{degenerate}. 
\end{defn}

\begin{defn}\label{def:canalizing}
A Boolean function $f(x_{1},\ldots, x_{n})$ is canalizing if there exists a variable $x_i$, a Boolean function $g(x_1,\ldots,x_{i-1},x_{i+1},\ldots,x_n)$ and $a,b \in \{0,1\}$ such that
$$f(x_1,\ldots,x_n) = 
 \begin{cases}
  b & \ \text{if} \  x_i = a, \\
  g\not\equiv b & \ \text{if} \ x_i \neq a,
  \end{cases}
$$
in which case $x_i$ is called a canalizing variable, the input $a$ is the canalizing input, and the output value $b$ when $x_i=a$ is the corresponding canalized output.
\end{defn}

Some authors, e.g., Kauffman in his original work on canalization~\cite{kauffman1974large}, do not require $g\not\equiv b$. In that case, constant functions are considered canalizing. Requiring $g\not\equiv b$ implies that only essential variables can be canalizing, a natural assumption as we will see later.

\begin{defn}\label{def:k_canalizing}\cite{he2016stratification}
A Boolean function $f(x_1,\ldots,x_n)$ is $k$-canalizing, where $1 \leq k \leq n$, with respect to the permutation $\sigma \in \mathcal{S}_n$, inputs $a_1,\ldots,a_k$ and outputs $b_1,\ldots,b_k$, if
\begin{equation}\label{eq:kcanalizing}f(x_{1},\ldots,x_{n})=
\left\{\begin{array}[c]{ll}
b_{1} & x_{\sigma(1)} = a_1,\\
b_{2} & x_{\sigma(1)} \neq a_1, x_{\sigma(2)} = a_2,\\
b_{3} & x_{\sigma(1)} \neq a_1, x_{\sigma(2)} \neq a_2, x_{\sigma(3)} = a_3,\\
\vdots  & \vdots\\
b_{k} & x_{\sigma(1)} \neq a_1,\ldots,x_{\sigma(k-1)}\neq a_{k-1}, x_{\sigma(k)} = a_k,\\
g\not\equiv b_k & x_{\sigma(1)} \neq a_1,\ldots,x_{\sigma(k-1)}\neq a_{k-1}, x_{\sigma(k)} \neq a_k,
\end{array}\right.\end{equation}
where $g = g(x_{\sigma(k+1)},\ldots,x_{\sigma(n)})$ is a Boolean function on $n-k$ variables. When $g$ is not canalizing, the integer $k$ is the canalizing depth of $f$ (as in \cite{layne2012nested}) and the variables $x_{\sigma(1)}, \ldots, x_{\sigma(k)}$ are called \emph{conditionally canalizing}~\cite{dimitrova2022revealing}. If, in addition, $g$ is not constant, it is called the core function of $f$. If $f$ is not canalizing, we set its canalizing depth $k=0$. That is, we define all Boolean functions to be $0$-canalizing.
\end{defn}

Any Boolean function has a unique extended monomial form, in which the variables are partitioned into different canalizing layers based on their importance or dominance.

\begin{theorem}~\cite{he2016stratification}\label{thm:he}
Every Boolean function $f(x_1,\ldots,x_n)\not\equiv 0$ can be uniquely written as 
\begin{equation}\label{eq:matts_theorem}
    f(x_1,\ldots,x_n) = M_1(M_2(\cdots (M_{r-1}(M_rp_C + 1) + 1)\cdots)+ 1)+ q,
\end{equation}

where each $M_i = \displaystyle\prod_{j=1}^{k_i} (x_{i_j} + a_{i_j})$ is a non-constant extended monomial, $k = \displaystyle\sum_{i=1}^r k_i$ is the canalizing depth, and $p_C$ is the \emph{core polynomial} of $f$. Each $x_i$ appears in exactly one of $\{M_1,\ldots,M_r,p_C\}$, and the only restrictions are the following ``exceptional cases'':
\begin{enumerate}[(i)]
    \item If $p_C\equiv 1$ and $r\neq 1$, then $k_r\geq 2$;
    \item If $p_C\equiv 1$ and $r = 1$ and $k_1=1$, then $q=0$.
\end{enumerate}
When $f$ is not canalizing (i.e., when $k=0$), $p_C = f$.
\end{theorem}

\begin{remark}~\label{rem:layer_output}
Any variable that is canalizing (independent of the values of other variables) appears in $M_1$, the first \emph{canalizing layer}. Any variable that ``becomes'' canalizing when excluding all variables from the first layer is in $M_2$, the second layer, etc. All conditionally canalizing variables appear in one of the layers $M_1,M_2,\ldots,M_r$. On the contrary, variables that never become canalizing appear in the core polynomial. 

Variables in the same layer may have different canalizing input values but they must have the same canalized output value since they set the function to the same value. This implies that the representation of a $k$-canalizing function as in Equation \ref{eq:kcanalizing} is generally not unique (since canalizing variables in the same layer may be reordered). However, the canalizing depth as well as the core function are uniquely defined, independent of the specific representation.
\end{remark}

\begin{example}
The Boolean function $f(x_1,x_2,x_3,x_4) = x_1 \vee \bar x_2 \vee (x_3 \oplus x_4)$ possesses one canalizing layer (since $x_1=1$ or $x_2=0$ both independently canalize $f$ to $1$), canalizing depth $2$, and core function $y \oplus z$. Note that $\oplus$ denotes addition modulo 2.
\end{example}


\section{Enumeration of non-degenerate Boolean functions with defined canalizing properties}

In this section, we count the number of Boolean functions in $n$ inputs with defined properties (canalizing depth, number of canalizing layers) \emph{and} a defined number of essential inputs. We first describe solutions to the much simpler task of stratifying all $n$-input Boolean functions by canalizing depth \emph{or} number of essential inputs.

\begin{defn} \label{def:n} 
Let $N(n)$ denote the number of Boolean functions in $n$ inputs. Let $N(n,m)$ denote the number of such functions with $m$ essential inputs. Further, let $N(n,m,k)$ denote the number of such functions with canalizing depth $k$. Finally, let $C(n,k)$ denote the number of $n$-input Boolean functions with canalizing depth $k$, i.e., not stratified by number of essential variables $m$.
\end{defn}

\begin{remark}
  Formulas for $C(n,k)$ are known~\cite{he2016stratification}. For example, 
    \begin{align*}
        C(n=0,k=0) &= 2,\\
        C(n=1,k=0) &= 2,\\
        C(n=1,k=1) &= 2,\\
        C(n=2,k=0) &= 4,\\
        C(n=2,k=1) &= 4,\\
        C(n=2,k=2) &= 8,\\
        C(n=3,k=0) &= 138,\\
        C(n=3,k=1) &= 30,\\
        C(n=3,k=2) &= 24,\\
        C(n=3,k=3) &= 64.\\
        \end{align*}
With these formulas, we can compute the probability that a random $n$-input Boolean function has canalizing depth $k$. As the number of inputs increases, canalization quickly becomes an exceedingly rare property (Fig.~\ref{fig:depth_vs_variables_n0to5}). This rarity makes it all the more striking that the vast majority of regulatory functions in published biological Boolean network models -- even those with $n \ge 5$ inputs -- are canalizing, and often even nested canalizing~\cite{kadelka2024meta}. 
\end{remark}

\begin{figure}
    \centering
    \includegraphics[width=0.6\linewidth]{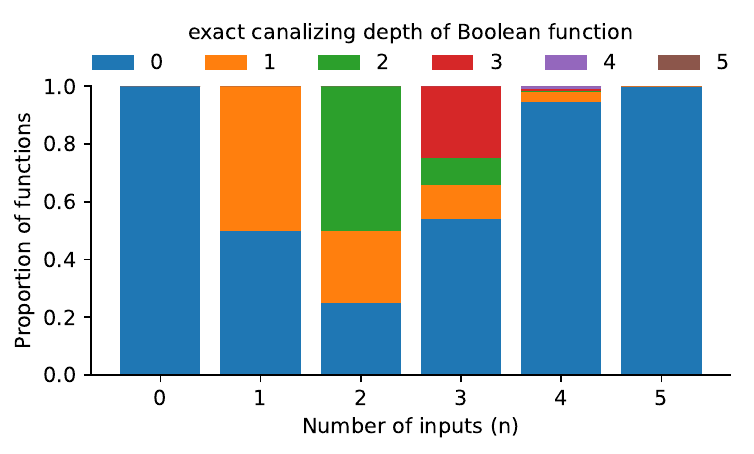}
    \caption{Proportion of $n$-input Boolean functions with a specific canalizing depth, computed using formulas for $C(n,k)$. Canalization becomes increasingly rare as $n$ increases, highlighting the need for exact formulas to study the prevalence of such functions.}
    \label{fig:depth_vs_variables_n0to5}
\end{figure} 

\begin{lemma}\label{lemma1}
By Definition~\ref{def:n}, we have 
\begin{align*}
    N(n) &= \sum_{m=0}^n N(n,m),\\
    N(n) &= \sum_{k=0}^n C(n,k),\\
    N(n,m) &= \sum_{k=0}^n N(n,m,k),\\
    C(n,k) &= \sum_{m=0}^n N(n,m,k).
\end{align*}
\end{lemma}

\begin{remark}\label{rem:N_m}
We know that there are $N(n) = 2^{2^n}$ different Boolean functions in $n$ inputs. The number of such functions with $m<n$ essential variables can be derived recursively.

For any $n\geq 0$, we have $N(n,m=0) = 2$, the constant functions $0$ and $1$. To derive $N(n,m)$ for $m\in\{1,\ldots,n-1\}$, we realize that there are $\binom nm$ choices to pick $m$ out of $n$ variables and that for each choice, we have $N(m,m)$ distinct non-degenerate functions. Thus,
$$N(n,m) = \binom{n}{m}N(m,m).$$
With this and Lemma~\ref{lemma1}, we can compute $N(n,n)$, the number of non-degenerate Boolean functions in $n$ inputs:
$$N(n,n) = N(n) - \sum_{m=0}^{n-1} N(n,m).$$
\end{remark}

This is a known sequence $N(0,0),N(1,1),\ldots = 2, 2, 10, 218, 64594, \ldots$~\cite{oeis}, which can also be computed directly using the inclusion-exclusion principle,
$$N(n,n) = \sum_{m=0}^n (-1)^{n-m} \binom nm 2^{2^m}.$$

Using similar ideas, we can stratify these numbers additionally by the canalizing depth $k$ to derive a recursive formula for $N(n,m,k)$.

\begin{theorem}\label{thm:main}
    The number of Boolean functions in $n\geq 0$ inputs with $m\leq n$ essential inputs and canalizing depth $k$, $N(n,m,k)$, can be computed recursively as follows
    $$N(n,m,k) = \begin{cases}
        0 & \ \text{if}\ m<k,\\
        2 & \ \text{if}\ m=k=0,\\
        \binom{n}{m}N(m,m,k) & \ \text{if}\ k\leq m < n,\\
        C(n,k) - \sum_{i=0}^{n-1} N(n,i,k) & \ \text{if}\ m=n,\\
    \end{cases}$$
    where $C(n,k)$ denotes the number of Boolean functions in $n$ inputs with canalizing depth $k$, provided in~\cite{he2016stratification}.
\end{theorem}

\begin{proof}
Since constant functions are not canalizing, we have $N(n,m,k) = 0$ if $m<k$ and specifically,
$$N(n,m=0,k) = \begin{cases} 2 & \ \text{if}\ k=0,\\
0 & \ \text{if}\ k>0.
\end{cases}
$$
As in Remark~\ref{rem:N_m}, we can derive $N(n,m,k)$ for $m\in\{1,\ldots,n-1\}$ recursively:
$$N(n,m,k) = \binom{n}{m}N(m,m,k).$$
With this, we can compute $N(n,n,k)$, the number of non-degenerate Boolean functions in $n$ inputs with canalizing depth $k$:
$$N(n,n,k) = C(n,k) - \sum_{m=0}^{n-1} N(n,m,k).$$
\end{proof}

\begin{example}
For $n=1$, we get $N(1,0,0) = 2$, the two constant functions, and $N(1,1,1)= 2$, the functions $x$ and $\bar x$, while $N(1,1,0) = C(1,0) - N(1,0,0) = 2-2 = 0$.

For $n=2$, we know that $N(2,2) = 10$ of the 16 Boolean functions are non-degenerate. Of the 6 degenerate functions, $N(2,0) = 2$ have no essential inputs and $N(2,1) = 4$ (the functions $x_1, \bar x_1, x_2, \bar x_2$) have one essential input, which is also canalizing. Accordingly, Theorem~\ref{thm:main} yields 
\begin{align*}
    N(2,0,0) &= 2,\\
    N(2,1,0) &= 2N(1,1,0) = 0,\\
    N(2,1,1) &= 2N(1,1,1) = 4.
\end{align*}
Moreover, 8 of the 10 non-degenerate 2-input functions are nested canalizing ($b + (x_1 + a_1)(x_2 + a_2)$ for arbitrary binary choices of canalizing input values $a_1$ and $a_2$ and canalized output value $b$) and 2 are 0-canalizing (the XOR and the XNOR function). As expected, we get 
\begin{align*}
    N(2,2,0) &= C(2,0) - N(2,0,0) - N(2,1,0) = 4 - 2 - 0 = 2,\\
    N(2,2,1) &= C(2,1) - N(2,0,1) - N(2,1,1) = 4 - 0 - 4 = 0,\\
    N(2,2,2) &= C(2,2) - N(2,0,2) - N(2,1,2)  = 8 - 0 - 0 = 8.
\end{align*}

For $n=3$, we have
\begin{align*}
    N(3,3,0) &= C(3,0) - N(3,0,0) - N(3,1,0) - N(3,2,0) = 138 - 2 - 0 - 6 = 130,\\
    N(3,3,1) &= C(3,1) - N(3,0,1) - N(3,1,1) - N(3,2,1) = 30 - 0 - 6 - 0 = 24,\\
    N(3,3,2) &= C(3,2) - N(3,0,2) - N(3,1,2) - N(3,2,2) = 24 - 0 - 0 - 24 = 0,\\
    N(3,3,3) &= C(3,3) - N(3,0,3) - N(3,1,3) - N(3,2,3) = 64 - 0 - 0 - 0 = 64.
\end{align*}
\end{example}

Theorem~\ref{thm:main} allows for the exact computation of the prevalence of $k$-canalization among non-degenerate Boolean functions with $n$ inputs (Fig.~\ref{fig:depth_vs_essential_variables_n0to5}).
A comparison with Fig.~\ref{fig:depth_vs_variables_n0to5} reveals substantial differences for functions with few inputs.
Because such low-input functions constitute the majority of regulatory update rules in published biological network models -- the mean in-degree across 122 networks is 2.56~\cite{kadelka2024meta} -- this result represents an important theoretical advance, providing a precise quantification of the extent to which canalization is over-represented in biological systems.
Moreover, only closed-form expressions permit reliable estimates of the abundance of $k$-canalization for Boolean functions with $n \ge 5$ inputs (Fig.~\ref{fig:depth_vs_essential_variables_n0to5}B), where the double-exponential growth of the state space, $2^{2^n}$, renders any simulation-based approach computationally infeasible.

\begin{figure}
    \centering
    \includegraphics[width=0.6\linewidth]{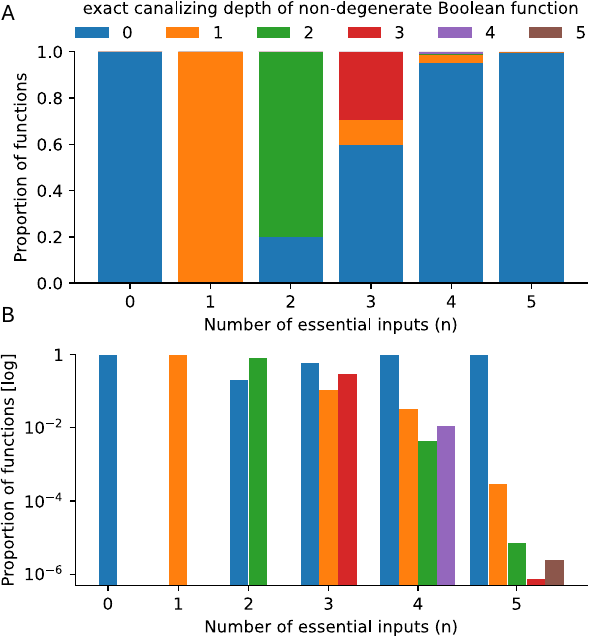}
    \caption{Proportion on a (A) linear and (B) log scale of $n$-input non-degenerate Boolean functions with a specific canalizing depth, computed using the formula for $N(n,m=n,k)$ from Theorem~\ref{thm:main}. These exact formulas enable an exploration of the prevalence of canalization even in functions with higher inputs ($n\geq 5$) where it is very rare.}
    \label{fig:depth_vs_essential_variables_n0to5}
\end{figure} 

\begin{remark}
    In the previous example, we see that $N(1,1,0) = N(2,2,1) = N(3,3,2) = 0$ although $C(1,0), C(2,1), C(3,1) \neq 0$. This hints at a general rule $N(n,m=n,k=n-1) = 0$, implying that non-degenerate functions cannot have canalizing depth $n-1$. This is a known fact~\cite{he2016stratification}: If a non-degenerate function had canalizing depth $n-1$, the subfunction $g$ (see Definition~\ref{def:k_canalizing}), which is evaluated if all $n-1$ conditionally canalizing variables receive their non-canalizing input value, would depend only on a single variable, i.e., we would have $g = x_i$ or $g=\bar x_i$, meaning $g$ is canalizing as well and the canalizing depth must be $n$. In other words, NCFs can only have $r\in\{1,2,\ldots,n-1\}$ canalizing layers because the last layer contains at least two variables.
\end{remark}

\begin{remark}
The quantity $\sum_{k=1}^{n} N(n, m = n, k = k)$ is of particular interest, as it gives the total number of non-degenerate canalizing Boolean functions with $n$ inputs. The ratio
$$P_{\text{canalizing}}(n) := \frac{\sum_{k=1}^n N(n,m=n,k=k)}{N(n,m=n)}$$ 
thus represents the prevalence of canalizing functions among all non-degenerate $n$-input functions. Similarly, $N(n, m = n, k = n)$ yields the number of non-degenerate NCFs. This quantity is identical to $C(n, k = n)$, since by definition all NCFs are non-degenerate. We let $$P_{\text{NCF}}(n) := \frac{N(n, m = n, k = n)}{N(n,m=n)} = \frac{C(n, k = n)}{N(n,m=n)}$$ denote the prevalence of NCFs among non-degenerate functions.

To quantify the effect of ignoring degeneracy, we compare $P_{\text{canalizing}}(n)$ and $P_{\text{NCF}}(n)$ to
$$\tilde P_{\text{canalizing}}(n) := \frac{\sum_{k=1}^n C(n,k=k)}{N(n)} \quad \text{and} \quad \tilde P_{\text{NCF}}(n) := \frac{C(n,k=n)}{N(n)},$$
which measure the prevalence of canalizing and nested canalizing functions among all $n$-input Boolean functions. The discrepancy between $\tilde{P}_{*}(n)$ and ${P}_{*}(n)$, illustrated in Fig.~\ref{fig:ratio}, captures the bias introduced when degenerate functions are not excluded. While the absolute difference $|\tilde{P}_{*}(n) - P_{*}(n)|$ naturally decreases as canalization itself becomes rarer for larger $n$, it is more informative to consider their $\log_2$-fold-change
$$\Delta_*(n) = \log_2(\frac{\tilde{P}_{*}(n)}{{P}_{*}(n)}),$$
which expresses how many powers of two the naïve estimate $\tilde{P}_{*}(n)$ over- or under-states the true prevalence among non-degenerate functions. 
For instance, the value of $\Delta_{\text{canalizing}}(1) = \Delta_{\text{NCF}}(1) = -1$
indicates that correctly accounting for degeneracy doubles the frequency of canalizing and nested canalizing functions in $n=1$ input (Fig.~\ref{fig:ratio}). For any $n$, the frequency of NCFs is always at least minimally higher when considering only non-degenerate functions (because $N(n,m=n) < N(n)$ for all $n$). 
Interestingly, the frequency of non-degenerate canalizing functions is higher than that of all canalizing functions for $n\in\{1,2\}$ but lower for $n\in\{3,4,5\}$.
\end{remark}

\begin{figure}
    \centering
    \includegraphics[width=0.5\linewidth]{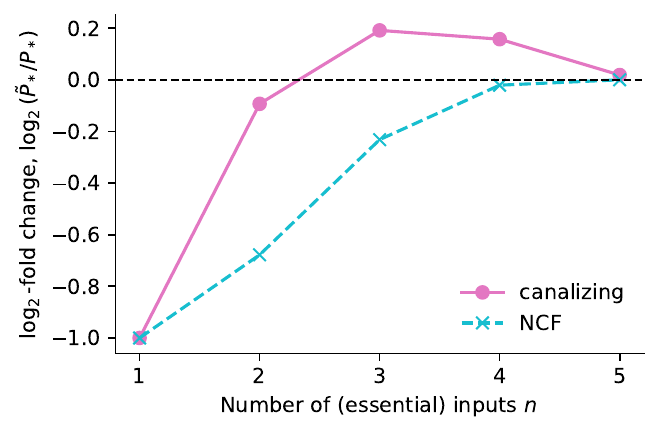}
    \caption{Relative bias introduced by neglecting degeneracy when estimating the prevalence of canalizing and nested-canalizing Boolean functions. For each number of inputs $n$, the plotted values represent the $\log_2$-fold change $\Delta_{*}(n)=\log_{2}(\tilde P_{*}(n)/P_{*}(n))$, which quantifies how many powers of two the naïve estimate $\tilde P_{*}(n)$ (based on all Boolean functions) over- or under-represents the true prevalence among non-degenerate functions. Positive values correspond to an apparent enrichment of canalization when degeneracy is ignored.}
    \label{fig:ratio}
\end{figure}

Boolean $n$-input functions have also been classified by canalizing depth $k$ \emph{and} number of canalizing layers $r\leq k$~\cite{kadelka2017influence}. 

\begin{defn}
    Let $C(n,k,r)$ denote the number of $n$-input Boolean functions with canalizing depth $k$ and $r\leq k$ canalizing layers. Further, let $N(n,m,k,r)$ denote the number of such functions with $m$ essential inputs.
\end{defn}

Formulas for $C(n,k,r)$ were derived in~\cite{dimitrova2022revealing}. Using the same arguments as in Theorem~\ref{thm:main}, we can thus derive the number of such functions that are non-degenerate ($N(n,n,k,r)$), and more generally, the number of such functions that have a defined number of essential inputs.

\begin{cor}\label{cor:main}
    The number of Boolean functions in $n\geq 0$ inputs with $m\leq n$ essential inputs and canalizing depth $k$, $N(n,m,k)$, can be computed recursively as follows
    $$N(n,m,k,r) = \begin{cases}
        0 & \ \text{if}\ m<k\ \text{or}\ k<r,\\
        2 & \ \text{if}\ m=k=r=0,\\
        \binom{n}{m}N(m,m,k,r) & \ \text{if}\ k\leq m < n,\\
        C(n,k,r) - \sum_{i=0}^{n-1} N(n,i,k,r) & \ \text{if}\ m=n,\\
    \end{cases}$$
    where $C(n,k,r)$ denotes the number of Boolean functions in $n$ inputs with canalizing depth $k$ and $r$ canalizing layers, provided in~\cite{dimitrova2022revealing}.
\end{cor}

\section{Discussion}
We have derived exact recursive formulas that enumerate Boolean functions jointly by the number of essential variables and the canalizing depth, thereby determining the number of non-degenerate canalizing functions for any number of inputs $n$. 
These results extend previous enumerations of canalizing and nested canalizing functions by incorporating degeneracy, an often-overlooked but biologically meaningful property that distinguishes functions with fully redundant inputs from those in which every variable is essential and contributes to the output.

While the rarity of canalization among random Boolean functions is well established, our results refine this understanding by providing exact prevalence estimates that explicitly exclude degenerate functions, thereby defining a more accurate null model for comparison with biological logic. Yet empirical analyses show that most regulatory rules in biological network models -- even those involving five or more inputs -- are canalizing or even nested canalizing~\cite{daniels2018criticality,kadelka2024meta}. 
The formulas derived here therefore provide an improved quantitative benchmark for measuring the degree of over-representation of canalization in biological systems. 
This comparison highlights canalization not as a statistical artifact of random logic but as an evolved design principle that promotes robustness and functional stability.

Beyond their biological relevance, these enumeration results are of intrinsic mathematical interest. 
They establish an exact correspondence between structural constraints on Boolean functions and the combinatorial growth of function space, offering insight into how hierarchical control and redundancy interact. 
As the space of Boolean functions expands double-exponentially with $n$, direct enumeration quickly becomes intractable (practically for $n\geq 5$), making recursive and closed-form approaches indispensable for understanding the combinatorial architecture of function space.
This study lays the groundwork for future work aimed at obtaining closed-form expressions, extending the analysis to multistate or biased functions, and exploring how combinatorial structure constrains dynamics in complex logical and regulatory systems.

\section*{Acknowledgments}
This research was funded in part by the U.S. National Science Foundation (NSF), grant numbers DMS-2424632 and DMS-2451973.

\section*{Data and Code Availability}
All code used to generate the figures and numerical results is available at 
\url{https://github.com/ckadelka/nondegenerate-canalization}.

\bibliographystyle{unsrt}
\bibliography{ref}

\end{document}